\documentclass[a4paper,11pt]{amsart}
\usepackage{amssymb}
\usepackage{amsmath}
\usepackage[cyr]{aeguill}
\usepackage{cmap}
\usepackage{hyperxmp}
\usepackage{todonotes}
\usepackage[pdfdisplaydoctitle=true,
            colorlinks=true,
            urlcolor=blue,
            citecolor=blue,
            linkcolor=blue,
            pdfstartview=FitH,
            pdfpagemode= UseNone,
            bookmarksnumbered=true]{hyperref}
\newtheorem{thm}{Theorem}[section]
\newtheorem{cor}[thm]{Corollary}
\newtheorem{lem}[thm]{Lemma}
\newtheorem{prop}[thm]{Proposition}

\theoremstyle{definition}
\newtheorem{defn}[thm]{Definition}
\newtheorem{hyp}[thm]{Presupposition}

\theoremstyle{remark}
\newtheorem{rem}{Remark}

\numberwithin{equation}{section}

\newcommand{\NN}{\mathbb{N}}
\newcommand{\RR}{\mathbb{R}}
\newcommand{\ZZ}{\mathbb{Z}}
\newcommand{\Circle}{\mathbb{S}^{1}}

\newcommand{\DiffS}{\mathrm{Diff}^{\infty}(\mathbb{S}^{1})}

\newcommand{\VectS}{\mathrm{Vect}(\mathbb{S}^{1})}
\newcommand{\CS}{\mathrm{C}^{\infty}(\mathbb{S}^{1})}
\newcommand{\Homeo}{\mathrm{Homeo}^{+}(\mathbb{S}^{1})}

\newcommand{\D}[1]{\mathcal{D}^{#1}(\mathbb{S}^{1})}
\newcommand{\HH}[1]{H^{#1}(\mathbb{S}^{1})}

\newcommand{\norm}[1]{\left\Vert#1\right\Vert}
\newcommand{\abs}[1]{\left\vert#1\right\vert}

\newcommand{\op}[1]{\mathbf{op}\left(#1\right)}

\DeclareMathOperator{\diam}{diam} %
\hypersetup{
    pdftitle={Geodesic Completeness for Sobolev $H^{s}$-metrics on the Diffeomorphism Group of the Circle}, 
    pdfauthor={J. Escher and B. Kolev}, 
    pdfsubject={MSC 2010: 58D05, 35Q53}, 
    pdfkeywords={Euler equation, diffeomorphism group, Sobolev metrics of non-integral order}, 
    pdflang=EN 
    }
\begin{document}

\title[Geodesic Completeness]{Geodesic Completeness for Sobolev $H^{s}$-metrics on the Diffeomorphism Group of the Circle}%

\author{Joachim Escher}
\address{Institute for Applied Mathematics, University of Hanover, D-30167 Hanover, Germany}
\email{escher@ifam.uni-hannover.de}

\author{Boris Kolev}
\address{Aix Marseille Universit\'{e}, CNRS, Centrale Marseille, I2M, UMR 7373, 13453 Marseille, France}
\email{boris.kolev@math.cnrs.fr}

\subjclass{Euler equation; diffeomorphism group; fractional Sobolev metrics}%

\date{\today}%
\begin{abstract}
We prove that the weak Riemannian metric induced by the fractional Sobolev norm $H^s$ on the diffeomorphism group of the circle is geodesically complete, provided that $s>3/2$.
\end{abstract}
\maketitle

\section{Introduction}
\label{sec:introduction}

The interest in right-invariant metrics on the diffeomorphism group of the circle started when it was discovered by Kouranbaeva~\cite{Kou1999} that the Camassa--Holm equation~\cite{CH1993} can be recast as the Euler equation of the right-invariant metric on $\DiffS$ induced by the $H^{1}$ Sobolev inner product on the corresponding Lie algebra $\CS$. The well-posedness of the geodesics flow for the right-invariant metric induced by the $H^{k}$ inner product was obtained by Constantin and Kolev~\cite{CK2003}, for $k \in \NN $, $k \ge 1$, following the pioneering work of Ebin and Marsden~\cite{EM1970}. These investigations have been extended to the case of fractional order Sobolev spaces $H^{s}$ with $s\in \RR_+$, $s \ge 1/2$ by Escher and Kolev~\cite{EK2012}. The method used to establish local existence of geodesics is to extend the metric and its spray to the Hilbert approximation $\D{q}$ (the Hilbert manifold of diffeomorphisms of class $H^{q}$) and then to show that the (extended) spray is smooth. This was proved to work in \cite{EK2012} for $s \ge 1/2$ provided we choose $q >3/2$ and $q\ge 2s$. The well-posedness on $\DiffS$ follows as $q\to \infty$ from a regularity preserving result of the geodesic flow.

A Riemannian metric is \emph{strong} if at each point it induces a topological isomorphism between the tangent space and the cotangent space. It is \emph{weak} if it defines merely an injective linear mapping between the tangent space and the cotangent space. Note that on $\DiffS$ only weak metrics exist. Furthermore we also mention that the extended metric on $\D{q}$ is not strong but only weak as soon as $q>2s$.

On a Banach manifold equipped with a strong metric, the geodesic semi-distance induced by the metric is in fact a distance \cite{Lan1999}. This is no longer true for weak metrics. It was shown by Bauer, Bruveris, Harms, and Michor~\cite{BBHM2013} that this semi-distance identically vanishes for the $H^{s}$ metric if $0 \le s \le 1/2$, whereas it is a distance for $s > 1/2$. This distance is nevertheless probably not complete on $T\D{q}$. Indeed, although for a strong metric topological completeness implies geodesic completeness, this is generally not true for a weak metric. Finally, we recall that the metric induced by the $H^1$-norm (or equivalently by $A := I-D^{2}$) is not geodesically complete, c.f. \cite{CE2000}. The main result of this paper is the following.

\begin{thm}\label{thm:main}
Let $s>3/2$ be given. Then the geodesic flow on $\D{q}$ for $q \ge 2s + 1$ and on $\DiffS$, respectively, is complete for the weak Riemannian metric
induced by the $\HH{s}$-inner product.
\end{thm}

Completeness results for groups of diffeomorphisms on $\RR^{n}$ have been studied in \cite{TY2005} and in \cite{MM2013}. In both papers,
stronger conditions on $s$ had been presupposed: Compared to our setting $s$ has to be larger than $7/2$ in \cite{TY2005} and an \emph{integer}
larger than $2$ in \cite{MM2013}, respectively. Additionally, in the work of \cite{TY2005, MM2013}, the phenomenon that the diffeomorphisms of an
orbit with finite extinction time may degenerate in the sense of the remarks following Corollary~\ref{cor:blow-up} is not reported on.

Let us briefly give an outline of the paper. In Section~\ref{sec:righ-invariant-metrics}, we introduce basic facts on right-invariant metrics on $\DiffS$ and we recall a well-posedness result for related geodesic flows. In Section~\ref{sec:distance}, we introduce a complete metric structure on suitable Banach approximations of $\DiffS$, which allows us to describe the precise blow-up mechanism of finite time geodesics. This is the subject matter of Section~\ref{sec:blow-up}. In Section~\ref{sec:global_solutions}, we prove our main result, Theorem~\ref{thm:main}. In Appendix~\ref{app:mollifiers}, we recall the material on Friedrichs mollifier that have been used throughout the paper.

\section{Right-invariant metrics on $\DiffS$}
\label{sec:righ-invariant-metrics}

Let $\DiffS$ be the group of all smooth and orientation preserving diffeomorphism on the circle. This group is naturally equipped with a \emph{Fr\'{e}chet manifold} structure; it can be covered by charts taking values in the \emph{Fr\'{e}chet vector space} $\CS$ and in such a way that the change of charts are smooth mappings (a smooth atlas with only two charts may be constructed, see for instance~\cite{GR2007}).

Since both the composition and the inversion are smooth for this structure we say that $\DiffS$ is a \emph{Fr\'{e}chet-Lie group}, c.f. \cite{Ham1982}. Its Lie algebra, $\VectS$, is the space of smooth vector fields on the circle. It is isomorphic to $\CS$ with the Lie bracket given by
\begin{equation*}
    [u,v] = u_{x}v - uv_{x}.
\end{equation*}

From an analytic point of view, the Fr\'{e}chet Lie group $\DiffS$ may be viewed as an inverse limit of \emph{Hilbert manifolds}. More precisely, recall that the Sobolev space $\HH{q}$ is defined as the completion of $\CS$ for the norm
\begin{equation*}
  \norm{u}_{\HH{q}}:= \left( \sum_{n \in \ZZ}(1 + n^{2})^{q}\abs{\hat{u}_{n}}^{2} \right)^{1/2},
\end{equation*}
where $q\in \RR^{+}$ and where $\hat u_n$ stands for the $n-$th Fourier coefficient of $u\in L^{2}(\Circle)$. Let $\D{q}$ denote the set of all orientation preserving homeomorphisms $\varphi$ of the circle $\Circle$, such that both $\varphi$ and $\varphi^{-1}$ belong to the fractional Sobolev space $\HH{q}$. For $q > 3/2$, $\D{q}$ is a \emph{Hilbert manifold} and a \emph{topological group} \cite{EM1970}. It is however not a \emph{Lie group} because neither composition, nor inversion in $\D{q}$ are \emph{smooth}, see again \cite{EM1970}. We have
\begin{equation*}
    \DiffS = \bigcap_{q > \frac{3}{2}} \D{q}.
\end{equation*}

\begin{rem}\label{rem:trivial_bundle}
Like any Lie group, $\DiffS$ is a \emph{parallelizable} manifold:
\begin{equation*}
  T\DiffS \sim \DiffS \times \CS.
\end{equation*}
What is less obvious, however, is that $T\D{q}$ is also a \emph{trivial bundle}. Indeed, let
\begin{equation*}
  \mathfrak{t}: T\Circle \to \Circle \times \RR
\end{equation*}
be a \emph{smooth trivialisation} of the tangent bundle of $\Circle$. Then
\begin{equation*}
  T\D{q} \to \D{q}\times \HH{q}, \qquad \xi \mapsto \mathfrak{t} \circ \xi
\end{equation*}
is a \emph{smooth vector bundle isomorphism} (see~\cite[p.~107]{EM1970}).
\end{rem}

A \emph{right-invariant} metric on $\DiffS$ is defined by an inner product on the Lie algebra $\VectS = \CS$. In the following we assume that this inner product is given by
\begin{equation*}
   \langle u,v \rangle = \int_{\Circle} (Au)v \, dx,
\end{equation*}
where $A:\CS \to \CS$ is a $L^{2}$-symmetric, positive definite, invertible \emph{Fourier multiplier} (i.e. a continuous linear operator on $\CS$ which commutes with $D:=d/dx$). For historical reasons going back to Euler's work \cite{Eul1765a}, $A$ is called the \emph{inertia operator}.

By translating the above inner product, we obtain an inner product on each tangent space $T_{\varphi} \DiffS$
\begin{equation}\label{eq:right-invariant-metric}
   \langle \eta, \xi \rangle_{\varphi} = \langle \eta\circ\varphi^{-1}, \xi\circ\varphi^{-1} \rangle_{id} = \int_{\Circle} \eta (A_{\varphi}\xi) \varphi_{x} \,dx,
\end{equation}
where $\eta, \xi \in T_{\varphi} \DiffS$ and $A_{\varphi} = R_{\varphi} \circ A \circ R_{\varphi^{-1}} $, and $R_{\varphi}(v):=v\circ\varphi$. This defines a smooth weak Riemannian metric on $\DiffS$.

This weak Riemannian metric admits the following \emph{geodesic spray}\footnote{A Riemannian metric on a manifold $M$ defines a smooth function on $TM$, given by half the square norm of a tangent vector. The corresponding Hamiltonian vector field on $TM$, relatively to the pullback of the canonical symplectic structure on $T^{*}M$ is called the \emph{geodesic spray}.}
\begin{equation}\label{eq:geodesic-spray}
   F : (\varphi,v)\mapsto \left(\varphi,v,v, S_{\varphi}(v)\right)
\end{equation}
where
\begin{equation*}
  S_{\varphi}(v):=\left(R_{\varphi}\circ S\circ R_{\varphi^{-1}} \right)(v),
\end{equation*}
and $S$ is a quadratic operator on the Lie algebra given by:
\begin{equation*}
  S(u):= A^{-1}\left\{ [A,u]u_{x}-2(Au)u_{x}\right\}.
\end{equation*}

A {\em{geodesic}} is an integral curve of this second order vector field, that is a solution $(\varphi,v)$ of
\begin{equation}\label{eq:geodesic-equations}
\left\{
\begin{aligned}
    \varphi_{t} &= v,
    \\
    v_{t} &= S_{\varphi}(v),
\end{aligned}
\right.
\end{equation}
Given a geodesic $(\varphi,v)$, we define the {\em{Eulerian velocity}} as
\begin{equation*}
  u := v\circ\varphi^{-1}.
\end{equation*}
Then $u$ solves
\begin{equation}\label{eq:Euler-equation}
  u_{t} = -A^{-1}\left[ u(Au)_{x} + 2(Au)u_{x} \right],
\end{equation}
called the \emph{Euler equation} defined by the inertia operator $A$.

\begin{rem}
When $A$ is a \emph{differential operator} of order $r \ge 1$ then the quadratic operator
\begin{equation*}
  S(u)= A^{-1}\left\{ [A,u]u_{x}-2(Au)u_{x}\right\}
\end{equation*}
is of order $0$ because the commutator $[A,u]$ is of order not higher than $r-1$. One might expect, that for a larger class of operators $A$, the quadratic operator $S$ to be of order $0$ and consequently the second order system \eqref{eq:geodesic-equations} can be viewed as an ODE on $T\D{q}$.
\end{rem}

\begin{defn}
A Fourier multiplier $A=\op{a(k)}$ with symbol $a$ is of \emph{order} $r\in \RR$ if there exists a constant $C >0$ such that
\begin{equation*}
  \abs{a(k)} \le \,C \left(1+k^{2}\right)^{r/2},
\end{equation*}
for every $k \in \ZZ$. In that case, for each $q \ge r$, the operator $A$ extends to a bounded linear operator from $\HH{q}$ to $H^{q-r}(\Circle)$. In this paper we only consider symmetric operators, i.e. $a(k)\in\mathbb{R}$ for all $k\in\mathbb{Z}$.
\end{defn}

When $A$ is a differential operator of order $r \ge 1$, the map
\begin{equation}\label{eq:conjugate-mapping}
  \varphi \mapsto A_{\varphi}, \quad \D{q} \to \mathcal{L}(\HH{q},\HH{q-r})
\end{equation}
is smooth (it is in fact \emph{real analytic}) for $q > 3/2$ and $q \ge r$. Indeed, in this case $A_{\varphi}$ is a linear differential operator with coefficients consisting of polynomial expressions of $1/\varphi_{x}$ and of the derivatives of $\varphi$ up to order $r$. Unfortunately, this argument \emph{does not apply} to a general \emph{Fourier multiplier} $A=\op{p(k)}$. In that case, even if $A$ extends to a bounded linear operator from $\HH{q}$ to $\HH{q-r}$, one cannot conclude directly that the mapping $\varphi \mapsto A_{\varphi}$ is smooth, because the mapping
\begin{equation*}
  \varphi \mapsto R_{\varphi}, \quad \D{q} \to \mathcal{L}(\HH{\sigma},\HH{\sigma})
\end{equation*}
is \emph{not even continuous}\footnote{The map $(\varphi,u) \mapsto u\circ \varphi$ is however continuous but not differentiable.}, for any choice of $\sigma\in [0,q]$.

Let us now precisely formulate the conditions that will be required on the inertia operator subsequently.

\begin{hyp}\label{hyp:inertia_operator}
  The following conditions will be assumed on the inertia operator $A$:
  \begin{enumerate}
    \item[(a)] $A = \op{a(k)}$ is a Fourier multiplier of order $r \ge 1$, or equivalently, $a(k) = \mathcal{O}(\abs{k}^{r})$;\\
    \item[(b)] For all $q \ge r$, $A: \HH{q} \to \HH{q-r}$ is a bounded isomorphism, or equivalently, for all $k \in \ZZ$, $a(k) \ne 0$ and $1/a(k) = \mathcal{O}(\abs{k}^{-r})$;\\
    \item[(c)] For each $q >3/2$ with $q\ge r$, the mapping
    \begin{equation*}
      \varphi \mapsto A_{\varphi}, \quad \D{q} \to \mathcal{L}(\HH{q},\HH{q-r})
    \end{equation*}
    is smooth.
  \end{enumerate}
\end{hyp}

In~\cite{EK2012} we have specified conditions on the symbol of $A$ which guarantee that $A$ satisfies presupposition \ref{hyp:inertia_operator}. Particularly, inertia operators of the form of Bessel potentials, i.e.
\begin{equation*}
  \Lambda^{2s} := \op{ \left( 1+ k^{2} \right)^{s} },
\end{equation*}
which generate the inner product of the fractional order Sobolev space $\HH{s}$
\begin{equation*}
  (u,v)\mapsto \langle\Lambda^{s} u\vert\Lambda^{s} v\rangle_{L^{2}},\quad u,\,v\in\HH{s},
\end{equation*}
meet these conditions, provided that $s \ge 1/2$.

If the conditions~\ref{hyp:inertia_operator} are satisfied, then expression~\eqref{eq:right-invariant-metric} defines a smooth, weak Riemannian metric on $\D{q}$, provided that $q>3/2$ and $q\ge r$. Moreover, it can be shown that the spray $F$ defined by equation~\eqref{eq:geodesic-spray} extends to a smooth vector field $F_{q}$ on $T\D{q}$, which is the geodesic spray of the metric, c.f. \cite[Theorem 3.10]{EK2012}. In that case, the \emph{Picard-Lindel\"{o}f Theorem} on the Banach manifold $T\D{q}$ ensures that, given any initial data $(\varphi_{0},v_{0})\in T\D{q}$, there is a unique \emph{non-extendable solution} $(\varphi,v)$ of \eqref{eq:geodesic-equations}, defined on a maximal interval $I_q(\varphi_{0},v_{0})$, satisfying the initial condition
\begin{equation*}
  (\varphi(0),v(0)) = (\varphi_{0},v_{0}).
\end{equation*}

A remarkable observation due to Ebin and Marsden (see \cite[Theorem 12.1]{EM1970}) states that, if the initial data $(\varphi_{0},v_{0})$ is smooth, then the maximal time interval of existence $I_q(\varphi_{0},v_{0})$ is independent of the parameter $q$. This is an essential ingredient in the proof of the local existence theorem for geodesics on $\DiffS$ (see \cite{EK2012}).

\begin{thm}\label{thm:nonextend_sol}
Suppose that presupposotion~\ref{hyp:inertia_operator} hold true. Then, given any $(\varphi_{0},v_{0})\in T\DiffS$, there exists a unique non-extendable solution
\begin{equation*}
  (\varphi, v)\in C^\infty(J,T\DiffS)
\end{equation*}
of \eqref{eq:geodesic-equations}, with initial data $(\varphi_{0},v_{0})$, defined on the maximal interval of existence $J_{max}=(t^{-},t^{+})$. Moreover, the solution depends smoothly on the initial data.
\end{thm}

As a corollary, we get well-posedness for the corresponding Euler equation~\eqref{eq:Euler-equation}.

\begin{thm}\label{thm:Euler eq}
Assume that the operator $A$ satisfies presupposition~\ref{hyp:inertia_operator}. Let $v_{0}\in\DiffS$ be given and denote by $J_{max}$ the maximal interval of existence for \eqref{eq:geodesic-equations} with the initial datum $(id_{\Circle},v_{0})$. Set $u:=v\circ\varphi^{-1}$. Then
$u\in C^\infty(J_{max},\CS)$ is the unique non-extendable solution of the Euler equation
\begin{equation}\label{eq:Euler initial}
\left\{
\begin{aligned}
  & u_{t} = -A^{-1}\left[ u(Au)_{x} + 2(Au)u_{x} \right],
  \\
  & u(0) = v_{0}.
\end{aligned}
\right.
\end{equation}
\end{thm}

It is also worth to recall that the metric norm along the flow is conserved.

\begin{lem}\label{lem:norm_conservation}
  Let $u$ be a solution to~\eqref{eq:Euler-equation} on the time interval $J$, then
  \begin{equation}\label{eq:norm_conservation}
    \norm{u(t)}_{A} = \left(\int_{\Circle} (Au)u \, dx\right)^{1/2}
  \end{equation}
  is constant on $J$.
\end{lem}

\section{A complete metric structure on $\D{q}$}
\label{sec:distance}

We recall that in what follows, $\Circle$ is the unit circle of the complex plane and that $\DiffS$ and $\D{q}$ may be considered as subset of the set $C^{0}(\Circle,\Circle)$ of all continuous maps of the circle. Besides the Banach manifold $\D{q}$ may be covered by two charts (see \cite{EK2011} for instance). We let
\begin{equation*}
  d_{0}(\varphi_{1},\varphi_{2}) := \max_{x \in \Circle} \abs{\varphi_{2}(x) -\varphi_{1}(x)}
\end{equation*}
be the $C^{0}$-distance between continuous maps of the circle. Endowed with this distance $C^{0}(\Circle,\Circle)$ is a complete metric space. Let $\Homeo$ be the group of orientation preserving homeomorphisms of the circle. Equipped with the induced topology, $\Homeo$ is a topological group, and each right translation $R_{\varphi}$ is an isometry for the distance $d_{0}$.

\begin{defn}
Given $q > 3/2$, we introduce the following distance on $\D{q}$
\begin{equation*}
  d_{q}(\varphi_{1},\varphi_{2}) := d_{0}(\varphi_{1},\varphi_{2}) + \norm{\varphi_{1x} - \varphi_{2x}}_{H^{q-1}} + \norm{1/\varphi_{1x} - 1/\varphi_{2x}}_{\infty}.
\end{equation*}
\end{defn}

\begin{lem}\label{lem:bound_diff}
Let $q>3/2$ be given and assume that $B$ is a bounded subset of $(\D{q},d_{q})$. Then
\begin{equation*}
  \inf_{\varphi\in B} \left(\min_{y\in\Circle} \varphi_{x}(y) \right) > 0.
\end{equation*}
\end{lem}

\begin{proof}
Let $M:=\diam B$, fix $\varphi_{0}\in B$ and put $\varepsilon:= 1/(M+\Vert1/\varphi_{0x}\Vert_{\infty})$. By hypothesis $M<\infty$, thus $\varepsilon >0$. Assume now by contradiction that
\begin{equation*}
  \inf_{\varphi\in B} \left(\min_{y\in\Circle}\varphi_{x}(y)\right) = 0.
\end{equation*}
Then there is a $\varphi_{1}\in B$ such that $\min_{y\in\Circle} \varphi_{1x}(y) < \varepsilon$. Using
\begin{equation*}
  \norm{1/\varphi_{1x}}_{\infty} = \max_{y\in\Circle}\left(\frac{1}{\varphi_{1x}(y)}\right) = \left(\min_{y\in\Circle} \varphi_{1x}(y)\right)^{-1} > \frac{1}{\varepsilon},
\end{equation*}
we find by the definition of $\varepsilon$ the contradiction:
\begin{equation*}
  M \ge d(\varphi_{1},\varphi_{0}) \ge \norm{1/\varphi_{1x} - 1/\varphi_{0x}}_{\infty} \ge \norm{1/\varphi_{1x}}_{\infty} - \norm{1/\varphi_{0x}}_{\infty} > M,
\end{equation*}
which completes the proof.
\end{proof}

\begin{prop}\label{prop:complete_metric_space}
Let $q>3/2$. Then $(\D{q},d_{q})$ is a complete metric space and its topology is equivalent to the Banach manifold topology on $\D{q}$.
\end{prop}

\begin{proof}
Let $\tau$ be the Banach manifold topology on $\D{q}$ and $\tau_{d}$ be the metric topology. Then
\begin{equation*}
  id : (\D{q},\tau) \to (\D{q},\tau_{d})
\end{equation*}
is continuous because $\varphi \mapsto \varphi^{-1}$ is a homeomorphism of $\D{q}$ (equipped with the manifold topology) and the fact that
\begin{equation*}
  d_{0}(\varphi_{1},\varphi_{2}) \lesssim \norm{\tilde{\varphi}_{1}-\tilde{\varphi_{2}}}_{H^{q}}
\end{equation*}
if $\tilde{\varphi}_{1}$ and $\tilde{\varphi}_{2}$ are lifts of $\varphi_{1}$ and $\varphi_{2}$ respectively. Conversely
\begin{equation*}
  id : (\D{q},\tau_{d}) \to (\D{q},\tau)
\end{equation*}
is continuous because given $\varphi_{0}$, there exists $\delta > 0$ such that if $d_{0}(\varphi_{0},\varphi) < \delta$, then $\varphi$ belongs to the same chart as $\varphi_{0}$ and in a local chart we have
\begin{equation*}
  \norm{\tilde{\varphi} - \tilde{\varphi_{0}}}_{H^{q}} \lesssim d_{q}(\varphi,\varphi_{0}).
\end{equation*}
This shows the equivalence of the two topologies.

Let now $(\varphi_{n})$ be a Cauchy sequence for the distance $d_{q}$. We observe first that $(\varphi_{n})$ converges in $C^{0}(\Circle,\Circle)$ to a map $\varphi$, that this map is $C^{1}$ and that $\varphi_{nx} \to \varphi_{x}$ in $\HH{q-1}$, because for $n$ large enough, all $\varphi_{n}$ belong to a same chart.
Invoking Lemma~\ref{lem:bound_diff}, we know that
\begin{equation*}
  \inf_{n\in\NN} \left(\min_{y\in\Circle}\varphi_{nx}(y)\right) > 0.
\end{equation*}
This implies that $\varphi_{x}>0$ and hence that $\varphi$ is a $C^{1}$-diffeomorphism of class $H^{q}$, and finally that $d_q(\varphi_{n},\varphi) \to 0$.
\end{proof}

\begin{lem}\label{lem:path_estimate}
Let $\varphi \in C^{1}(I,\D{q})$ be a path in $\D{q}$ and let $v:=\varphi_{t}$ be its velocity. Then
\begin{equation*}
  d_q(\varphi(t),\varphi(s)) \lesssim \abs{t-s} \max_{[s,t]} \norm{v}_{H^{q}} \left(1 + \max_{[s,t]} \norm{1/\varphi_{x}}^{2}_{\infty}\right)
\end{equation*}
for all $t,s \in I$.
\end{lem}

\begin{proof}
Let $\widetilde{\varphi} \in C^{1}(I,H^{q}(\RR))$ be a lift of the path $\varphi$. Given $s,t \in I$ with $s<t$, we have first
\begin{equation}\label{eq:estimate_1}
\begin{split}
  d_{0}(\varphi(t),\varphi(s)) & \lesssim \norm{\widetilde{\varphi}(t) - \widetilde{\varphi}(s)}_{\infty}\\
    & \le \int_{s}^{t} \norm{\varphi_{t}(\tau)}_{\infty} d\tau
    \lesssim \abs{t-s} \max_{[s,t]} \norm{v}_{H^{q}}.
\end{split}
\end{equation}
Next, we have
\begin{equation*}
  \varphi_{x}(t) - \varphi_{x}(s) = \int_{s}^{t} \varphi_{tx}(\tau) d\tau,
\end{equation*}
in $\HH{q-1}$ and hence
\begin{equation}\label{eq:estimate_2}
  \norm{\varphi_{x}(t) - \varphi_{x}(s)}_{H^{q-1}} \le \int_{s}^{t} \norm{\varphi_{tx}(\tau)}_{H^{q-1}} d\tau \le \abs{t-s} \max_{[s,t]} \norm{v}_{H^{q}}.
\end{equation}
Finally we have
\begin{equation}\label{eq:estimate_3}
\begin{split}
  \norm{1/\varphi_{x}(t) - 1/\varphi_{x}(s)}_{\infty} & \le \left(\max_{[s,t]} \norm{1/\varphi_{x}}_{\infty}\right)^{2} \int_{s}^{t} \norm{\varphi_{tx}(\tau)}_{\infty} \\
    & \lesssim \abs{t-s} \max_{[s,t]} \norm{v}_{H^{q}} \left(\max_{[s,t]} \norm{1/\varphi_{x}}_{\infty}\right)^{2}.
\end{split}
\end{equation}
Fusing \eqref{eq:estimate_1}, \eqref{eq:estimate_2}, and \eqref{eq:estimate_3} completes the proof.
\end{proof}

\section{The blow-up scenario for geodesics}
\label{sec:blow-up}

In the sequel a bounded set in $\D{q}$ will always mean \emph{bounded relative to the distance $d_{q}$} and a bounded set in $T\D{q}= \D{q}\times \HH{q}$ will mean \emph{bounded relative to the product distance}
\begin{equation*}
  d_{q}(\varphi_{1},\varphi_{2}) + \norm{v_{1} - v_{2}}_{H^{q}}.
\end{equation*}
The main result of this section is the following.

\begin{thm}\label{thm:Spray_bounded}
Let $q>3/2$ be given with $q \ge r$. Then the geodesic spray
\begin{equation*}
  F_{q} : (\varphi,v) \mapsto (v,S_{\varphi}(v))
\end{equation*}
is bounded on bounded sets of $\D{q} \times \HH{q}$.
\end{thm}

The proof of this theorem is based on Lemma~\ref{lem:Rphi_bounded}, which is itself a corollary of the following estimates obtained in~\cite[Appendix B]{EK2012}.
\begin{equation}\label{eq:Rphi_estimate_firstcase}
  \norm{R_{\varphi}}_{\mathcal{L}(\HH{\rho},\HH{\rho})} \le C^{1}_{\rho}\left(\norm{1/\varphi_{x}}_{L^{\infty}},\norm{\varphi_{x}}_{L^{\infty}}\right),
\end{equation}
for $0 \le \rho \le 1$,
\begin{equation}\label{eq:Rphi_estimate_secondcase}
  \norm{R_{\varphi}}_{\mathcal{L}(\HH{\rho},\HH{\rho})} \le C^{2}_{\rho}\left(\norm{1/\varphi_{x}}_{L^{\infty}},\norm{\varphi_{x}}_{H^{q-1}}\right),
\end{equation}
for $0 \le \rho \le 2$,
\begin{equation}\label{eq:Rphi_estimate_thirdcase}
  \norm{R_{\varphi}}_{\mathcal{L}(\HH{\rho},\HH{\rho})} \le C^{3}_{\rho}\left(\norm{1/\varphi_{x}}_{L^{\infty}},\norm{\varphi_{x}}_{L^{\infty}}\right) \norm{\varphi_{x}}_{H^{\rho-1}},
\end{equation}
for $3/2 < \rho \le 3$,
\begin{equation}\label{eq:Rphi_estimate_fourthcase}
  \norm{R_{\varphi}}_{\mathcal{L}(\HH{\rho},\HH{\rho})} \le C^{4}_{\rho}\left(\norm{1/\varphi_{x}}_{L^{\infty}},\norm{\varphi_{x}}_{H^{\rho-2}}\right) \norm{\varphi_{x}}_{H^{\rho-1}},
\end{equation}
for $\rho > 5/2$, and
\begin{equation}\label{eq:phi_inverse_estimate}
   \norm{(\varphi^{-1})_{x}}_{H^{\rho-1}} \lesssim C^{5}_{\rho}(\norm{1/\varphi_{x}}_{\infty}, \norm{\varphi_{x}}_{H^{\rho-1}}),
\end{equation}
for $\rho > 3/2$, where $C^{k}_{\rho}$ is a positive, continuous function on $(\RR^{+})^{2}$, for $k=1,\dots,5$.

\begin{lem}\label{lem:Rphi_bounded}
Let $q > 3/2$ and $ 0 \le \rho \le q$ be given. Then the mappings
\begin{equation*}
  \varphi \mapsto R_{\varphi}, \quad \D{q} \to \mathcal{L}(\HH{\rho},\HH{\rho})
\end{equation*}
and
\begin{equation*}
  \varphi \mapsto R_{\varphi^{-1}}, \quad \D{q} \to \mathcal{L}(\HH{\rho},\HH{\rho})
\end{equation*}
are bounded on bounded subsets of $\D{q}$.
\end{lem}

\begin{proof}[Proof of Theorem~\ref{thm:Spray_bounded}]
Recall that $S_{\varphi}(v) = R_{\varphi} \circ S \circ R_{\varphi^{-1}}$ where
\begin{equation*}
  S(u):= A^{-1}\left\{ [A,u]u_{x}-2(Au)u_{x}\right\}.
\end{equation*}
In particular, $S_{\varphi}(v)$ is quadratic in $v$ and
\begin{equation*}
  \norm{S_{\varphi}(v)}_{H^{q}} \le \norm{R_{\varphi}}_{\mathcal{L}(H^{q},H^{q})} \norm{S}_{\mathcal{L}(H^{q}\times H^{q},H^{q})} \norm{R_{\varphi^{-1}}}_{\mathcal{L}(H^{q},H^{q})} ^{2}\norm{v}_{H^{q}}^{2}.
\end{equation*}
Now, $S$ is a bounded bilinear operator and $R_{\varphi}$ and $R_{\varphi^{-1}}$ are bounded on bounded subsets of $\D{q}$ by Lemma~\ref{lem:Rphi_bounded}. This completes the proof.
\end{proof}

Our next goal is to study the behaviour of geodesics which do not exists globally, i.e. $t^{+}<\infty$ or $t^{-}>-\infty$. We have the following result, which is a consequence of Theorem~\ref{thm:Spray_bounded}.

\begin{cor}\label{cor:blow-up}
Assume that presupposition~\ref{hyp:inertia_operator} are satisfied and let
\begin{equation*}
  (\varphi, v)\in C^\infty((t^{-},t^{+}),T\D{q})
\end{equation*}
denote the non-extendable solution of the geodesic flow~\eqref{eq:geodesic-equations}, emanating from
\begin{equation*}
  (\varphi_{0},v_{0})\in T\D{q}.
\end{equation*}
If $t^{+}<\infty$, then
\begin{equation*}
   \lim_{t\uparrow t^{+}} \left[ d_q(\varphi_{0},\varphi(t))+\norm{v(t)}_{H^{q}} \right] = +\infty.
\end{equation*}
A similar statement holds true if $t^{-}>-\infty$.
\end{cor}

\begin{proof}
Suppose that $t^{+}<\infty$ and set
\begin{equation*}
  f(t) := d_{q}(\varphi_{0},\varphi(t))+\norm{v(t)}_{H^{q}}
\end{equation*}
where $(\varphi(t),v(t))\in T\D{q}$ is the solution of \eqref{eq:geodesic-equations} at time $t\in(t^{-},t^{+})$, emanating from $(\varphi_{0},v_{0})$.

(i) Note first that $f$ cannot be bounded on $[0,t^{+})$. Otherwise, the spray $F_{q}(\varphi(t),v(t))$ would be bounded on $[0,t^{+})$ by Theorem~\ref{thm:Spray_bounded}. In that case, given any sequence $(t_{k})$ in $[0,t^{+})$ converging to $t^{+}$, we would conclude, invoking Lemma~\ref{lem:path_estimate}, that $(\varphi(t_{k}))$ is a Cauchy sequence in the complete metric space $(\D{q},d_{q})$. Similarly, we would conclude that the sequence $(v(t_{k}))$ is a Cauchy sequence in the Hilbert space $\HH{q}$. Then, by the Picard-Lindel\"{o}f theorem, we would deduce that the solution could be extended beyond $t^{+}$, which would contradict the maximality of $t^{+}$.

(ii) We are going to show now that
\begin{equation*}
  \lim_{t\nearrow t^{+}} f(t) = + \infty.
\end{equation*}
If this was wrong, then we would have
\begin{equation*}
  \liminf_{t\nearrow t^{+}} f < +\infty \quad \text{and} \quad \limsup_{t\nearrow  t^{+}} f = +\infty .
\end{equation*}
But then, using the continuity of $f$, we could find $r >0$ and two sequences $(s_{k})$ and $(t_{k})$ in $[0,t^{+})$, each converging to $t^{+}$, with
\begin{equation*}
 s_{k} < t_{k}, \quad f(s_{k}) =r, \quad f(t_{k}) =2r
\end{equation*}
and such that
\begin{equation*}
  f(t) \le 2r, \quad \forall t \in \bigcup_{k} [s_{k},t_{k}].
\end{equation*}
However, by Theorem~\ref{thm:Spray_bounded}, we can find a positive constant $M$ such that
\begin{equation*}
  \norm{S_{\varphi}(v)}_{H^{q}} \le M,
\end{equation*}
for all $(\varphi,v)\in T\D{q}$ satisfying
\begin{equation*}
  d_{q}(\varphi_{0},\varphi) + \norm{v}_{H^{q}} \le 2r.
\end{equation*}
We would get therefore, using again Lemma~\ref{lem:path_estimate}, that
\begin{equation*}
  r = f(t_{k}) - f(s_{k}) \le C \abs{t_{k}-s_{k}}, \quad \forall k\in \mathbb{N},
\end{equation*}
for some positive constant $C$, which would lead to a contradiction and completes the proof.
\end{proof}

Assume that $t^{+}<\infty$. Then Corollary~\ref{cor:blow-up} makes it clear that there are only two possible blow-up scenarios: either the solution $(\varphi(t),v(t))$ becomes large in the sense that
\begin{equation*}
  \lim_{t\to t^{+}} \left(\norm{\varphi_{x}(t)}_{H^{q-1}} + \norm{v(t)}_{H^{q}}\right)=\infty,
\end{equation*}
or the family of diffeomorphisms $\{\varphi(t)\,;\, t\in(t^{-},t^{+})\}$ becomes singular in the sense that
\begin{equation*}
  \lim_{t\to t^{+}} \left(\min_{x\in\Circle}\{\varphi_{x}(t,x)\}\right)=0.
\end{equation*}

It is however worth emphasizing that the blow-up result in Corollary~\ref{cor:blow-up} only represents a necessary condition. Indeed, for $A=I-D^{2}$, i.e. for the Camassa--Holm equation the precise blow-up mechanism is known (see~\cite{CE2000}): a classical solution $u$ blows up in finite time if and only if
\begin{equation}\label{cond_global}
  \lim_{t\to t^{+}} \left(\min_{x\in\Circle}\{u_{x}(t,x)\}\right)=-\infty,
\end{equation}
which is somewhat weaker than blow up in $\HH{2}$. Since it is known that any (classical) solution to the Camassa--Holm equation preserves the $H^1$ norm and thus stays bounded, one says that the blow up occurs as a \emph{wave breaking}. Note also that
\begin{equation*}
  u_{x}(t,x) = v_{x}\circ\varphi(t,x)\cdot\frac{1}{\varphi_{x}(t,x)} \quad \hbox{for} \quad (t,x)\in (t^{-},t^{+})\times\Circle.
\end{equation*}
Hence in the case of a wave breaking, either $\abs{v_{x}}$ becomes unbounded or $v_{x}$ becomes negative and $\varphi_{x}$ tends to $0$ as $t\uparrow t^{+}$.

On the other hand there are several evolution equations, different from the Camassa--Holm equation, e.g. the \emph{Constantin-Lax-Majda equation} \cite{CLM1985,Wun2010}, which corresponds to the case $A=\mathcal{H}D$, where $\mathcal{H}$ denotes the Hilbert transform, cf. \cite{EKW2012} for which the blow up mechanism is much less understood and so far no sharper results than blow up in $\HH{1+\sigma}$ for any $\sigma>1/2$ or pointwise vanishing of $\varphi_{x}$ seem to be known.

\section{Global solutions}
\label{sec:global_solutions}

Throughout this section, we suppose that the inertia operator $A$ satisfies conditions~\ref{hyp:inertia_operator}. We fix some $q \ge r + 1$, and we let
\begin{equation}\label{eq:sol-geodesic}
  (\varphi,v)\in C^\infty(J,T\D{q})
\end{equation}
be the unique solution of the Cauchy problem~\eqref{eq:geodesic-equations}, emanating from
\begin{equation*}
  (id_{\Circle},v_{0})\in T\D{q}
\end{equation*}
and defined on the \emph{maximal time interval} $J=(t^{-},t^{+})$. The corresponding solution $u=v\circ\varphi^{-1}$ of the Euler equation \eqref{eq:Euler initial} is a path
\begin{equation}\label{Euler-sol}
  u \in C^{0}(J,\HH{q})\cap C^{1}(J,\HH{q-1}),
\end{equation}
because
\begin{equation*}
  (\varphi,v) \mapsto v \circ \varphi^{-1}, \qquad \D{q} \times \HH{q} \to \HH{q-1},
\end{equation*}
is $C^{1}$ for $q>3/2$ (see~\cite[Corollary B.6]{EK2012}). Moreover, since $A$ is of order $r$, the \emph{momentum} $m(t):=Au(t)$ is defined as a path
\begin{equation}\label{eq:Euler-Poincare-sol}
  m \in C^{0}(J,\HH{q-r})\cap C^{1}(J,\HH{q-r-1}).
\end{equation}
It satisfies the \emph{Euler-Poincar\'{e} equation}
\begin{equation}\label{eq:Euler-Poincare-equation}
  m_{t} = -m_{x}u - 2mu_{x} \quad \text{in} \quad C(J,L^2(\Circle)).
\end{equation}
We will prove that the geodesic $(\varphi(t),v(t))$ is defined for all time, as soon as $u_{x}$ is bounded below, independently of a particular choice of the inertia operator $A$, provided that $r \ge 2$.

\begin{rem}
Global solutions in $\HH{q}$ ($q > 3/2$) of the Camassa--Holm equation, which corresponds to the special case where the inertia operator $A = 1 - D^{2}$, have been studied in \cite{Mis2002}. It was established there, that $u(t)$ is defined on $[0,\infty)$ provided $\norm{u}_{C^{1}}$ is bounded \cite[Theorem 2.3]{Mis2002}. A similar argument was used in \cite{MZ2009} to establish existence of solutions of the Euler equation for the inertia operator $A = (1 - D^{2})^{k}$, $k\ge 1$, for which $m(t)$ does not blow up in $L^{2}$.
\end{rem}

The main result of this section is the \emph{a priori} estimate contained in the following result.

\begin{thm}\label{thm:eulerian_velocity_bound}
Let $r \ge 2$ and $q \ge r + 1$ be given and let
\begin{equation*}
  u \in C^{0}(J,\HH{q})\cap C^{1}(J,\HH{q-1})
\end{equation*}
be the solution of~\eqref{eq:Euler-equation} with initial data $u_{0} \in \HH{q}$ on $J$. Let $I$ be some \emph{bounded} subinterval of $J$ and suppose that
\begin{equation*}
  \inf_{t \in I}\left(\min_{x\in\Circle}\{u_{x}(t,x)\}\right) > - \infty.
\end{equation*}
Then $\norm{u}_{H^{q}}$ is bounded on $I$.
\end{thm}

The approach used here is inspired by that of Taylor~\cite{Tay1991} and relies on \emph{Friedrichs mollifiers} (see Appendix~\ref{app:mollifiers}). It requires also the following commutator estimate due to Kato and Ponce \cite{KP1988} (see also~\cite{Tay2003}).

\begin{lem}\label{Kato-Ponce}
  Let $s >0$ and $\Lambda^{s} := \op{(1+k^{2})^{s/2}}$. If $u,v \in \HH{s}$, then
  \begin{equation}\label{eq:KP_estimate}
    \norm{\Lambda^{s}(uv)-u\Lambda^{s}(v)}_{L^{2}} \lesssim \norm{u_{x}}_{\infty} \norm{\Lambda^{s-1}v}_{L^{2}} + \norm{\Lambda^{s}u}_{L^{2}} \norm{v}_{\infty}
  \end{equation}
\end{lem}

\begin{proof}[Proof of Theorem~\ref{thm:eulerian_velocity_bound}]

(1) Let $m(t) = Au(t)$ for $t\in J$. Invoking \eqref{eq:Euler-Poincare-sol} and the fact that $q-r-1 \ge 0$, we conclude that the curve $[t \mapsto m(t)]$ belongs to $C^{1}(J,L^2(\Circle))$. Thus the Euler-Poincar\'{e} equation~\eqref{eq:Euler-Poincare-equation}
implies that
\begin{equation*}
  \frac{d}{dt}\norm{m}_{L^{2}}^{2} = -2 \left\langle m, m_{x}u + 2mu_{x}\right\rangle_{L^{2}}\quad \text{on}
\quad J.
\end{equation*}
Consequently, we get
\begin{equation*}
  \frac{d}{dt}\norm{m}_{L^{2}}^{2} \le -3 \min_{x\in\Circle}\{u_{x}(t,x)\} \norm{m}_{L^{2}}^{2},
\end{equation*}
and by virtue of Gronwall's lemma, we conclude that $\norm{m}_{L^{2}}$ is bounded on $I$. Recalling that $A^{-1}$ is a bounded operator from $L^{2}(\Circle)$ to $\HH{r}$, we see that $\norm{u}_{H^{r}}$ is bounded on $I$. This applies, in particular, to $\norm{u}_{H^{2}}$, because we assumed that $r \ge 2$.

(2) Our next goal is to derive an $H^1$ \textit{a priori} estimate for $m$. Since the curve $[t \mapsto m(t)]$ belongs merely to $C^{1}(J,\HH{q-r-1})$ and $q-r-1$ may be smaller than $1$, we need to replace it by the curve $t \mapsto J_{\varepsilon}m(t)$, where $J_{\varepsilon}$ is a \emph{Friedrichs' mollifier} with respect to the spatial variable in $\Circle$, cf. Appendix~\ref{app:mollifiers}. We note that $J_{\varepsilon}m \in C^1(J,C^\infty(\Circle))$. For this regularized curve $J_{\varepsilon}m$, we are going now to show that
\begin{equation}\label{eq:gronH1}
  \frac{d}{dt}\norm{J_{\varepsilon}m}_{H^{1}}^{2} \lesssim \norm{u}_{H^{2}} \norm{m}_{H^{1}}^{2},
\end{equation}
for $\varepsilon \in (0,1]$. To do so, note that
\begin{multline*}
  \frac{d}{dt}\norm{J_{\varepsilon}m}_{H^{1}}^{2} = -2 \int (J_{\varepsilon}m) (J_{\varepsilon}m_{x}u) - 4\int (J_{\varepsilon}m) (J_{\varepsilon}mu_{x}) \\
  -4\int (J_{\varepsilon}m_{x}) (J_{\varepsilon}mu_{xx}) -6\int (J_{\varepsilon}m_{x}) (J_{\varepsilon}m_{x}u_{x}) -2\int (J_{\varepsilon}m_{x}) (J_{\varepsilon}m_{xx}u).
\end{multline*}
Using Cauchy--Schwarz' inequality and Lemma~\ref{lem:J_norm_estimate}, the first four terms of the right hand-side can easily be bounded by $\norm{u}_{H^{2}} \norm{m}_{H^{1}}^{2}$, up to a positive constant independent of $\varepsilon$. The last term in the right hand-side can be rewritten as
\begin{equation*}
  \int (J_{\varepsilon}m_{x}) (uJ_{\varepsilon}m_{xx}) + \int (J_{\varepsilon}m_{x}) ([J_{\varepsilon},uD]m_{x}).
\end{equation*}
An integration by parts shows that the first term is bounded by $\norm{u_{x}}_{\infty} \norm{m}_{H^{1}}^{2}$. By Cauchy-Schwarz' inequality and Lemma~\ref{lem:J_commutator_estimate}, the same is true for the second term.

(3) Suppose now that $3/2 < \sigma \le q-r$. We are going to show that
\begin{equation}\label{eq:gronHsigma}
  \frac{d}{dt}\norm{J_{\varepsilon}m}_{H^{\sigma}}^{2} \lesssim \norm{u}_{H^{\sigma+1}} \norm{m}_{H^{\sigma}}^{2},
\end{equation}
for $\varepsilon \in (0,1]$. We have
\begin{equation*}
  \frac{d}{dt}\norm{J_{\varepsilon}m(t)}_{H^{\sigma}}^{2}
= -2 \left\langle \Lambda^{\sigma}J_{\varepsilon}m, \Lambda^{\sigma}J_{\varepsilon}(m_{x}u)\right\rangle_{L^{2}} - 4 \left\langle \Lambda^{\sigma}J_{\varepsilon}m, \Lambda^{\sigma}J_{\varepsilon}(mu_{x})\right\rangle_{L^{2}}.
\end{equation*}
Applying Cauchy-Schwarz' inequality, we first get
\begin{equation*}
  \left\langle \Lambda^{\sigma}J_{\varepsilon}m, \Lambda^{\sigma}J_{\varepsilon}(mu_{x})\right\rangle_{L^{2}} \le \norm{J_{\varepsilon}m}_{H^{\sigma}}\norm{J_{\varepsilon}(mu_{x})}_{H^{\sigma}}
\end{equation*}
and, by virtue of~\eqref{eq:FM1}, we have
\begin{equation*}
  \norm{J_{\varepsilon} m}_{H^{\sigma}} \norm{J_{\varepsilon}(mu_{x})}_{H^{\sigma}} \lesssim \norm{m}_{H^{\sigma}}\norm{mu_{x}}_{H^{\sigma}} \lesssim \norm{u}_{H^{\sigma+1}} \norm{m}_{H^{\sigma}}^{2},
\end{equation*}
uniformly in $\varepsilon$ (because $\HH{\sigma}$ is a multiplicative algebra as soon as $\sigma > 1/2$). Observing that $\Lambda^{\sigma}$ and $J_{\varepsilon}$ commute (see Appendix~\ref{app:mollifiers}), we have
\begin{multline}\label{eq:second_term}
  \left\langle \Lambda^{\sigma}J_{\varepsilon}m, \Lambda^{\sigma}J_{\varepsilon}(m_{x}u)\right\rangle_{L^{2}} = \int J_{\varepsilon}\left( u \Lambda^{\sigma} m_{x}\right) J_{\varepsilon}\Lambda^{\sigma}m \\
  + \int J_{\varepsilon}\left( [\Lambda^{\sigma},u] m_{x}\right) J_{\varepsilon}\Lambda^{\sigma}m.
\end{multline}
By virtue of Cauchy--Schwarz' inequality, \eqref{eq:FM1} and the Kato--Ponce estimate (Lemma~\ref{Kato-Ponce}), the second term in the right hand-side of~\eqref{eq:second_term} is bounded (up to a constant independent of $\varepsilon$) by
\begin{equation*}
  \norm{u}_{H^{\sigma}} \norm{m}_{H^{\sigma}}^{2},
\end{equation*}
because $\norm{m_{x}}_{\infty} \lesssim \norm{m}_{H^{\sigma}}$ for $\sigma > 3/2$. Introducing the operator $L := uD$, the first term in the right hand-side of~\eqref{eq:second_term} can be written as
\begin{equation*}
  \int (J_{\varepsilon}L\Lambda^{\sigma}m)(J_{\varepsilon}\Lambda^{\sigma}m) = \int (LJ_{\varepsilon}\Lambda^{\sigma}m)(J_{\varepsilon}\Lambda^{\sigma}m) + \int ([J_{\varepsilon},L]\Lambda^{\sigma}m)(J_{\varepsilon}\Lambda^{\sigma}m).
\end{equation*}
We have first
\begin{equation*}
  \int (LJ_{\varepsilon}\Lambda^{\sigma}m)(J_{\varepsilon}\Lambda^{\sigma}m) = \frac{1}{2} \int \{(L+ L^{*})J_{\varepsilon}\Lambda^{\sigma}m\}(J_{\varepsilon}\Lambda^{\sigma}m).
\end{equation*}
But, since $L+ L^{*} = -u_{x}I$, we get
\begin{equation*}
  \int  \{ (L+ L^{*})J_{\varepsilon}\Lambda^{\sigma}m \}(J_{\varepsilon}\Lambda^{\sigma}m) \lesssim \norm{u_{x}}_{\infty} \norm{m}_{H^{\sigma}}^{2}.
\end{equation*}
Now, using Cauchy--Schwarz' inequality and Lemma~\ref{lem:J_commutator_estimate}, we have
\begin{equation*}
  \int ([J_{\varepsilon},L]\Lambda^{\sigma}m)(J_{\varepsilon}\Lambda^{\sigma}m) \lesssim \norm{u_{x}}_{\infty}\norm{m}_{H^{\sigma}}^{2}.
\end{equation*}
Combining these estimates, we obtain finally
\begin{equation*}
  \frac{d}{dt}\norm{J_{\varepsilon}m(t)}_{H^{\sigma}}^{2} \lesssim \norm{u}_{H^{\sigma+1}} \norm{m}_{H^{\sigma}}^{2}.
\end{equation*}

(4) If either $\sigma=1$ or $\sigma > 3/2$, we integrate \eqref{eq:gronH1} or \eqref{eq:gronHsigma}, respectively, over $[0,t]$ to get
\begin{equation*}
  \norm{J_{\varepsilon}m(t)}_{H^{\sigma}}^{2} \le \norm{J_{\varepsilon}m(0)}_{H^{\sigma}}^{2} + C \sup_{\tau \in [0,t]}\norm{u(\tau)}_{H^{\sigma+1}} \int_{0}^{t} \norm{m(\tau)}_{H^{\sigma}}^{2}\,d\tau, \quad t \in J,
\end{equation*}
for some positive constant $C$ (independent of $\varepsilon$). Again, letting $\varepsilon \to 0$ and invoking~\eqref{eq:FM2} in combination with Gronwall's lemma, we conclude that $\norm{m(t)}_{H^{\sigma}}$ is bounded on $I$, as soon as $\norm{u(t)}_{H^{\sigma+1}}$ is. Therefore, using an inductive argument, we deduce that $\norm{u(t)}_{H^{q}}$ is bounded on $I$. This completes the proof.
\end{proof}

We next derive estimates on the flow map induced by time-dependent vector fields. These results are independent of the geodesic flow~\eqref{eq:geodesic-equations}. Therefore we formulate them in some generality. Note that on a general Banach manifold, the flow of a continuous vector field may not exist~\cite{Die1950}. However, in the particular case we consider here, we have the following result.

\begin{prop}[Ebin-Marsden, \cite{EM1970}]
Let $q > 5/2$ be given and let $u \in C^{0}\left(I, \HH{q}\right)$ be a time dependent $H^{q}$ vector field. Then its flow $t \to \varphi(t)$ is a $C^{1}$ curve in $\D{q}$.
\end{prop}

\begin{lem}\label{lem:flow_bound_C1}
  Let $u \in C^{0}\left( J, \HH{q}\right)$ be a time dependent vector field with $q > 3/2$. Assume that its associated flow $\varphi$ exists and that $\varphi\in C^{1}(J,\D{q})$. If $\norm{u_{x}}_{\infty}$ is bounded on any bounded subinterval of $J$, then $\norm{\varphi_{x}}_{\infty}$ and $\norm{1/\varphi_{x}}_{\infty}$ are bounded on any bounded subinterval of $J$.
\end{lem}

\begin{proof}
Let
\begin{equation*}
  \alpha(t) = \max_{x\in\Circle} \varphi_{x}(t)(x), \quad \text{and} \quad \beta(t) = \max_{x\in\Circle} 1/ \varphi_{x}(t)(x).
\end{equation*}
Note that $\alpha$ and $\beta$ are continuous functions. Let $I$ denote any bounded subinterval of $J$, and set
\begin{equation*}
  K = \sup_{t \in I} \norm{u_{x}(t)}_{\infty}.
\end{equation*}
From equation $\varphi_{t} = u \circ \varphi$, we deduce that
\begin{equation*}
  \varphi_{tx} = (u_{x} \circ \varphi ) \varphi_{x}, \quad \text{and} \quad \left( 1/\varphi_{x} \right)_{t}= -(u_{x}\circ \varphi) / \varphi_{x},
\end{equation*}
and therefore, we get
\begin{equation*}
  \alpha(t) \le \alpha(0) + K \int_{0}^{t}\alpha(s)\, ds
  \quad\text{and}\quad
  \beta(t) \le \beta(0) + K \int_{0}^{t}\beta(s)\, ds.
\end{equation*}
Thus the conclusion follows from Gronwall's lemma.
\end{proof}

\begin{lem}\label{lem:flow_bound_Hq}
Let $u\in C^{0}(J,\HH{q})$ with $q>3/2$ be a time-dependent vector field and assume that its associated flow $\varphi$ exists with $\varphi\in C^1(J,\D{q})$. If $\norm{u}_{H^{q}}$ is bounded on any bounded subinterval of $J$, then $\norm{\varphi_{x}}_{H^{q-1}}$ is bounded on any bounded subinterval of $J$.
\end{lem}

\begin{proof}
Let $I$ denote any bounded subinterval of $J$. For $0 \le \rho \le q-1$, we have
\begin{equation*}
  \frac{d}{dt}\norm{\varphi_{x}}_{H^{\rho}}^{2} = 2\left\langle (u\circ\varphi)_{x}, \varphi_{x}\right\rangle_{H^{\rho}} \lesssim \norm{(u\circ\varphi)_{x}}_{H^{\rho}} \norm{\varphi_{x}}_{H^{\rho}}.
\end{equation*}

(i) Suppose first that $1/2 < \rho \le 1$. Invoking~\eqref{eq:Rphi_estimate_firstcase}, we get
\begin{align*}
  \norm{(u\circ\varphi)_{x}}_{H^{\rho}}
  & \lesssim \norm{u_{x}\circ\varphi}_{H^{\rho}} \norm{\varphi_{x}}_{H^{\rho}}
  \\
  & \lesssim C_{\rho}^{1}\left(\norm{1/\varphi_{x}}_{L^{\infty}},\norm{\varphi_{x}}_{L^{\infty}}\right) \norm{u}_{H^{q}} \norm{\varphi_{x}}_{H^{\rho}}.
\end{align*}
Therefore, using the fact that $\norm{\varphi_{x}}_{\infty}$ and $\norm{1/\varphi_{x}}_{\infty}$ are bounded on $I$ by virtue of Lemma~\ref{lem:flow_bound_C1}, we conclude by Gronwall's lemma that $\norm{\varphi_{x}}_{H^{\rho}}$ is bounded on $I$, for $0 \le \rho \le 1$.

(ii) Suppose now that $1 \le \rho \le 2$. Invoking~\eqref{eq:Rphi_estimate_thirdcase}, we get
\begin{align*}
  \norm{(u\circ\varphi)_{x}}_{H^{\rho}}
  & \lesssim \norm{u\circ\varphi}_{H^{\rho+1}}
  \\
  & \lesssim C_{\rho+1}^{3}\left(\norm{1/\varphi_{x}}_{L^{\infty}},\norm{\varphi_{x}}_{L^{\infty}}\right) \norm{\varphi_{x}}_{H^{\rho}} \norm{u}_{H^{q}}.
\end{align*}
and we conclude again by Gronwall's lemma that $\norm{\varphi_{x}}_{H^{\rho}}$ is bounded on $I$, for $0 \le \rho \le 2$.

(iii) Suppose finally that $\rho \ge 3$. Invoking~\eqref{eq:Rphi_estimate_fourthcase}, we get
\begin{align*}
  \norm{(u\circ\varphi)_{x}}_{H^{\rho}}
  & \lesssim \norm{u\circ\varphi}_{H^{\rho+1}}
  \\
  & \lesssim C_{\rho + 1}^{4}\left(\norm{1/\varphi_{x}}_{L^{\infty}},\norm{\varphi_{x}}_{H^{\rho-1}}\right) \norm{\varphi_{x}}_{H^{\rho}} \norm{u}_{H^{q}}.
\end{align*}
and we conclude by an induction argument on $\rho$ that $\norm{\varphi_{x}}_{H^{\rho}}$ is bounded on $I$ for $0 \le \rho \le q-1$. This completes the proof.
\end{proof}

\begin{thm}\label{thm:global_sol}
Let $r \ge 2$ and $q \ge r+1$. Assume that conditions~\ref{hyp:inertia_operator} are satisfied and let
\begin{equation*}
  (\varphi, v)\in C^\infty((t^{-},t^{+}),T\D{q})
\end{equation*}
denote the non-extendable solution of the geodesic flow~\eqref{eq:geodesic-equations}, emanating from
\begin{equation*}
  (\varphi_{0},v_{0})\in T\D{q}.
\end{equation*}
If the Eulerian velocity $u = v \circ \varphi^{-1}$ satisfies the estimate
\begin{equation}\label{eq:ux-estimate}
  \inf_{t \in [0,t^{+})}\left(\min_{x\in\Circle}\{u_{x}(t,x)\}\right) > - \infty,
\end{equation}
then $t^{+}=\infty$. A similar statement holds for $t^{-}$.
\end{thm}

\begin{proof}
Assume that $t^{+} < \infty$ and that estimate~\eqref{eq:ux-estimate} holds. In view of Theorem~\ref{thm:eulerian_velocity_bound} we conclude that $\norm{u}_{H^{q}}$ is bounded on $[0,t^{+})$. By Lemma~\ref{lem:flow_bound_C1} we get furthermore that $\norm{\varphi_{x}}_{\infty}$, and $\norm{1/\varphi_{x}}_{\infty}$ are bounded on $[0,t^{+})$ and by Lemma~\ref{lem:flow_bound_Hq} we know that $\norm{\varphi_{x}}_{H^{q-1}}$ is bounded on $[0,t^{+})$. We obtain therefore that $\norm{v}_{H^{q}}$ is bounded on $[0,t^{+})$, by virtue of Lemma~\ref{lem:Rphi_bounded}. Therefore, we deduce that
\begin{equation*}
  d_{q}(\varphi_{0},\varphi(t)) + \norm{v(t)}_{H^{q}}
\end{equation*}
is bounded on $[0,t^{+})$. But this contradicts Corollary~\ref{cor:blow-up} which shows that
\begin{equation*}
  \lim_{t\uparrow t^{+}} \left[d(\varphi_{0},\varphi(t)) + \norm{v(t)}_{H^{q}}\right] = +\infty.
\end{equation*}
as soon as $t^{+} < + \infty$.
\end{proof}

Theorem~\ref{thm:main} follows from Theorem~\ref{thm:global_sol} and Lemma~\ref{lem:norm_conservation} in combination with Sobolev's embedding Theorem.

\begin{rem}
  The same conclusion holds for the weak Riemannian metric induced by any inertia operator $A$ of order $r > 3$ and satisfying presupposition~\ref{hyp:inertia_operator}, because then the norm
  \begin{equation*}
    \norm{u}_{A} := \langle Au,u\rangle_{L^{2}}
  \end{equation*}
  is equivalent to the $H^{r/2}$-norm.
\end{rem}

\appendix

\section{Friedrichs mollifiers}
\label{app:mollifiers}

\emph{Friedrichs mollifiers} were introduced by Kurt Otto Friedrichs in \cite{Fri1944}. We briefly recall the construction for periodic functions (see~\cite{Kra1992} for more details). Let $\rho$ be a nonnegative, even, smooth bump function of total weight $1$ and supported in $(-1/2,1/2)$. We set
\begin{equation*}
  \rho_{\epsilon} (x) := \frac{1}{\varepsilon} \rho\left(\frac{x}{\varepsilon}\right),
\end{equation*}
and define the \emph{Friedrichs' mollifer} $J_{\varepsilon}$ as the operator
\begin{equation*}
  J_{\varepsilon}u = \rho_{\epsilon} \ast u,
\end{equation*}
where $\ast$ denotes the convolution. Note that if $u \in L^{2}(\Circle)$, then $J_{\varepsilon}u \in \CS$ and that $J_{\varepsilon}$ is a bounded operator from $L^{2}(\Circle)$ to $\HH{q}$ for any $q \ge 0$.

The operator $J_{\varepsilon}$ is a Fourier multiplier. Thus it commutes with any other Fourier multiplier, in particular with the spatial derivative $D$. It commutes of course also with temporal derivative $\partial_{t}$ for functions depending on $(t,x)\in \mathbb{R}\times \Circle$. Note also that $J_{\varepsilon}$ is symmetric with respect to the $L^{2}$ scalar product. The main properties of $J_{\varepsilon}$ that have been used in this paper are the following.

\begin{lem}\label{lem:J_approximation}
Given $q \ge 0$ and $u\in \HH{q}$, then
\begin{equation}\label{eq:FM2}
  \norm{J_{\varepsilon}u - u}_{H^{q}} \to 0,\quad  \text{as}\ \varepsilon \to 0.
\end{equation}
\end{lem}

Lemma~\ref{lem:J_approximation} is a classical result. Its proof can be found in~\cite[Lemma 3.15]{Ada1975}), for instance.

\begin{lem}\label{lem:J_norm_estimate}
We have
\begin{equation}\label{eq:FM1}
  \norm{J_{\varepsilon}u}_{H^{q}} \lesssim \norm{u}_{H^{q}}, \quad \forall u \in \HH{q},
\end{equation}
\emph{uniformly} in $\varepsilon \in (0,1]$ and $q \ge 0$.
\end{lem}

The proof of Lemma~\ref{lem:J_norm_estimate} is a consequence of the following special case of \emph{Young's inequality} (\cite[Theorem 2.2, Chapter 1]{Pet1983})
\begin{equation}\label{eq:Young_inequality}
  \norm{f \ast u}_{L^{2}} \lesssim \norm{f}_{L^{1}} \norm{u}_{L^{2}},
\end{equation}
and the fact that $\Lambda^{q}$ and $J_{\varepsilon}$ commute.

Finally, we have been using the following commutator estimate on $[J_{\varepsilon}, uD]$.

\begin{lem}\label{lem:J_commutator_estimate}
  Let $u \in C^{1}(\Circle)$ and $m \in L^{2}(\Circle)$. Then
  \begin{equation*}
    \norm{J_{\varepsilon}(um_{x}) - uJ_{\varepsilon}(m_{x})}_{L^{2}} \lesssim \norm{u_{x}}_{\infty} \norm{m}_{L^{2}},
  \end{equation*}
  uniformly in $\varepsilon \in (0,1]$.
\end{lem}

\begin{proof}
Let $u \in C^{1}(\Circle)$. Note first that the linear operator
\begin{equation*}
  K_{\varepsilon} (m) := J_{\varepsilon}(um_{x}) - uJ_{\varepsilon}(m_{x}),
\end{equation*}
defined on $\CS$, is an integral operator with kernel
\begin{equation*}
  k_{\varepsilon}(x,y) = \frac{\partial}{\partial y} \left\{ (u(x)-u(y))\rho_{\varepsilon}(x-y)\right\}.
\end{equation*}
We have therefore
\begin{equation}\label{eq:FM3}
  K_{\varepsilon} (m) = - \rho_{\varepsilon} \ast (u_{x}m) - \int_{\Circle} \rho_{\varepsilon}^{\prime}(x-y) [u(x)-u(y)] m(y) \, dy.
\end{equation}
By virtue of Young's inequality~\eqref{eq:Young_inequality}, the $L^{2}$-norm of the first term of the right hand-side of~\eqref{eq:FM3} is bounded (up to some positive constant independent of $\varepsilon$) by
\begin{equation*}
  \norm{\rho_{\varepsilon}}_{L^{1}} \norm{u_{x}m}_{L^{2}} \le \norm{u_{x}}_{\infty} \norm{m}_{L^{2}},
\end{equation*}
because $\norm{\rho_{\varepsilon}}_{L^{1}}=1$. The $L^{2}$ norm of the second term of the right hand-side of~\eqref{eq:FM3} is bounded by
\begin{equation*}
  \left(\varepsilon \norm{u_{x}}_{\infty}\right) \norm{\rho_{\varepsilon}^{\prime}\ast m}_{L^{2}},
\end{equation*}
because the support of $\rho_{\varepsilon}$ is contained in $[-\varepsilon/2,\varepsilon/2]$. Using again Young's inequality~\eqref{eq:Young_inequality}, we get then
\begin{equation*}
  \norm{\rho_{\varepsilon}^{\prime}\ast m}_{L^{2}} \lesssim \norm{\rho_{\varepsilon}^{\prime}}_{L^{1}} \norm{m}_{L^{2}} \lesssim \frac{1}{\varepsilon} \norm{m}_{L^{2}},
\end{equation*}
because $\norm{\rho_{\varepsilon}^{\prime}}_{L^{1}} = \mathcal{O}(1/\varepsilon)$. This concludes the proof.
\end{proof}


\end{document}